\documentclass[11pt]{article}
\usepackage[letterpaper, margin=1in]{geometry} % 1-inch margins on letter-size paper
\usepackage{setspace} % For line spacing
\usepackage{titling}
\usepackage{amsfonts,amssymb,amsthm}
\usepackage[english]{babel}

\usepackage{mathrsfs}
\usepackage{enumitem}
\usepackage{amsthm}
\usepackage[super]{nth}
\usepackage{listings}
\usepackage{mathtools}
\usepackage{csquotes}
\usepackage{multicol,multirow}
\usepackage{fancyvrb}
\usepackage{makecell}
\usepackage{mdframed}
\usepackage[toc,page]{appendix}
\usepackage{mathdots}
\usepackage{empheq}
\usepackage{bm}
\usepackage{pgfplots}
\usepackage{enumitem}
\usepackage{float}
\usepackage{hyperref}
\usepackage{nicefrac} 
\usepackage{array}
\usepackage{blkarray}
\usepackage{amsmath}
\usepackage{graphicx}

\newtheorem{theorem}{Theorem}
\newtheorem{corollary}{Corollary}
\newtheorem{lemma}{Lemma}
\newtheorem{remark}{Remark}
\newtheorem{claim}{Claim}
\newtheorem{example}{Example}
\newcommand{\popcrit}{{\sc Crit-fPRI}}
\newcommand{\lev}{\mathrm{lev}}
\usepackage{tikz}
\usetikzlibrary{automata, positioning, calc, shapes, arrows, fit}

\newcommand{\C}{\mathcal{C}}

\newenvironment{claimproof}{\par\noindent\underline{Proof:}}{\leavevmode\unskip\penalty9999 \hbox{}\nobreak\hfill\quad\hbox{$\blacksquare$}}

\newcommand{\mplus}{(M-N)^+}
\newcommand{\nplus}{(N-M)^+}
\newcommand{\popmax}{{\sc MaxW-fPRI}}
\newcommand{\maxpri}{{\sc Max-fPRI}}
\newcommand{\w}{\omega}
\newcommand{\srti}{{\sc Max-fSRTI}}

\newcommand{\smti}{\textsc{Max-SMTI}}

\newcommand{\vote}{\mathrm{vote}}

\newcommand{\cgblocks}{$\gamma$-blocks}
\newcommand{\cgstab}{$\gamma$-stable}

\newcommand{\usrti}{{\sc Max-$\gamma$-fSRTI}}

\newcommand{\cgblocking}{$\gamma$-blocking}
\newcommand{\pbDef}[3]{
     \begin{center}
     \colorbox{lightgray}{\begin{minipage}{0.95\linewidth}%
    \textsc{#1}%\\[0.2ex]  
     \begin{itemize}[leftmargin=1.1cm]      
       \item[\textbf{In:}]  #2%\\[0.2ex]
       \item[\textbf{Out:}]  #3
     \end{itemize}
    \end{minipage}
    }
    \end{center}
  
}

\title{Extending Stable and Popular Matching Algorithms from Bipartite to Arbitrary Instances}

\author{
  Gergely Cs\'{a}ji$^{1,2}$\\[0.5em]
  \small $^1$Institute of Economics, Centre for Economic and Regional Studies (ELTE KRTK KTI), Hungary\\
  \small $^2$Department of Operations Research, Eötvös Loránd University (ELTE), Hungary
}

\begin{document}
\begin{titlepage}
	\maketitle
\begin{abstract}
We consider stable and popular matching problems in arbitrary graphs, which are referred to as stable roommates instances. 
We extend the $3/2$-approximation algorithm for the maximum size weakly stable matching problem to the roommates case, which solves a more than 20 year old open question of Irving and Manlove about the approximability of maximum size weakly stable matchings in roommates instances with ties \cite{manlove2002hard} and has nice applications for the problem of matching residents to hospitals in the presence of couples.
We also extend the algorithm that finds a maximum size popular matching in bipartite graphs in the case of strict preferences and the algorithm to find a popular matching among maximum weight matchings. 

While previous attempts to extend the idea of promoting the agents or duplicating the edges from bipartite instances to arbitrary ones failed, these results show that with the help of a simple observation, we can indeed bridge the gap and extend these algorithms.
\end{abstract}
\end{titlepage}

\section{Introduction}
Preference-based matching markets are a well-studied area in computer science, mathematics, and economics, with many practical applications. Research in this field began in 1962 
with the influential work of Gale and Shapley, who introduced the stable marriage problem and showed that a stable matching always exists and can be found in linear time \cite{gale1962college}. This discovery led to a large body of research on matching problems and their applications. For an overview, see Manlove \cite{manlove2013algorithmics}. Today, matching theory is used in many areas, including resident allocation, university admissions, kidney exchanges, job markets, and more, making it important for various decision-making tasks

A central question related to stable matchings is about finding a maximum size stable matching in a given instance with weak preferences (meaning ties are allowed). For bipartite instances, the problem is called \smti, which has been show to admit a 3/2-approximation algorithm by McDermid in 2009 \cite{Mcdermid09}. However, the question of approximability remained open for arbitrary instances ever since Irving and Manlove asked it in their seminal paper more than 20 years ago \cite{irving2002stable}. The approximablity of maximum size stable matchings in arbitrary (also called roommmates) instances was also posed as an open question at \cite{egres}. As our main result, we resolve this long-standing open question and show that the problem admits a 3/2-approximation algorithm just as in the bipartite special case. Somewhat surprisingly, this result has a nice application for the allocation of residents to hospitals in the presence of couples, which is also a central problem in the field, with applications in many country's national resident matching programs, such as the NRMP in the US \cite{nrmp}. 

Apart from stability, another central notion getting more and more attention recently is popularity, which was introduced by Gardenfors \cite{gardenfors1975match}. Intuitively, a popular matching is one that does not lose in a head-to-head election against any other matching. The advantages of popularity are that (i) it allows larger matchings than stable ones, while still preserving a nice global form of stability and (ii) it can naturally be made compatible with maximum size or maximum weight matchings, by restricting the candidate matchings. 

Using our central observation, we extend some previous algorithms for popular matchings to roommates instances. We give an algorithm to find a maximum size popular fractional-matching and also an algorithm that finds a popular solution among the maximum weight fractional-matchings.

\subsection{Related Work}\label{sec:related}
The stable marriage model - where the goal is to match two classes of agents in a way such that there
are no blocking edges, i.e. pairs of agents who mutually prefer each other to their partners - was introduced by Gale and Shapley \cite{gale1962college}. However, this work only considered complete and strict preferences.
The problem was first extended for the case where ties and incomplete preferences are allowed by Iwama et al.~\cite{IMMM99}, who showed the NP-hardness of \textsc{Max-SMTI}, which is the problem of finding a maximum size weakly stable matching in such an instance, where weakly stable refers to the fact that both agents need a strict improvement for a pair to block. Since then, several papers made progress to improve the approximability of the problem, e.g. Iwama et al. \cite{IMY07, IMY08} and Király \cite{Kiraly11}. Currently, the best ratio is $\frac{3}{2}$ by a polynomial time algorithm of McDermid \cite{Mcdermid09}, where the same ratio is attained by linear-time algorithms of Pauluch \cite{paluch2011faster, Paluch14} and Király \cite{kiraly2012linear, Kiraly13}.
%The current best lower bound is $\frac{33}{29}$ \cite{Yanagisawa07}. 
%Moreover, the integrality gap of a natural IP formulation for {\sc max-smti} is shown to be at least $\frac{3}{2}-o(1)$ in \cite{IMY14}, which rules out the possibility of further improvement by techniques such as rounding and primal-dual algorithms.
This $\frac{3}{2}$-approximation has been extended even to very general cases, including two-sided matroid constraints \cite{csaji2023approximation}, $\gamma$-stability \cite{csaji2023simple}, where we are given $\gamma_v^e$ numbers for each edge $e$ and vertex $v$ such that these numbers describe how much improvement $v$ needs to consider $e$ a blocking edge, and critical relaxed stability \cite{critical-ties-approx}, which is a generalization of \smti, where there is a set of critical agents, among whom we must match as much as possible and stability is relaxed in an appropriate way.

As for the inapproximability of \textsc{Max-SMTI}, Halldórssson et al.  \cite{halldorsson2002inapproximability} showed that it is NP-hard to approximate it within some constant factor. Yanagisawa \cite{yanagisawa2007approximation} and Halldórsson et al. \cite{yanigasawa2003improved} showed that assuming the Unique Games Conjecture, there is no $\frac{4}{3}-\varepsilon$-approximation for any $\varepsilon >0$. In a recent work, Dudycz et al. \cite{smti1.5inapprox} proved that assuming the strong Unique Games Conjecture or the Small Set Expansion Hypothesis, there cannot even be a $\frac{3}{2}-\varepsilon$-approximation algorithm for \textsc{Max-SMTI}.

The Stable Roommates problem with ties and incomplete preferences was studied by Irving and Manlove \cite{irving2002stable}. They showed that finding super stable matchings (where even a weak improvement is enough for a pair to block) can be done in polynomial time. They posed as an open question whether the maximum size of a weakly stable matching can be approximated better than 2 or not. Of course, as deciding if a stable roommates instance with ties admits a weakly stable matching is NP-hard \cite{irving2002stable}, the question can be interpreted in the following sense. By a result of Tan \cite{tan1991necessary}, we know that half-integral weakly stable matchings (also called weakly stable partitions) always exist and can be found in polynomial time, even in the case of multigraphs \cite{cechlarova2005stable}. Therefore, the question of whether we can approximate the maximum size of a weakly stable half-matching (or a weakly stable fractional-matching) better than 2 is still a very relevant open question. %Here, a 2-approximation is trivial, as any weakly stable matching is a 2-approximation, which has already been observed by Irving and Manlove in their paper \cite{irving2002stable}.
We show that both these problems admit a $\frac{3}{2}$-approximation algorithm, which is best possible under the strong Unique Games Conjecture, given the inapproximabilty results for the bipartite case.

Popular matchings have been studied by many papers recently. Intuitively, a matching is popular, if for any other matching, the number of agents who improve is at most the number of agents who get worse. It has been known since Gardenfors' seminal paper in 1975 \cite{gardenfors1975match} that stable matchings are popular, and hence popular matchings always exist in bipartite instances. Kavitha \cite{kavitha2014size} gave algorithms to find a maximum size popular matching, and a matching that is popular among the maximum size matchings. The ideas of her algorithms were similar to Király's, that is, they used Gale-Shapley type proposal-rejection algorithms, where agents could have multiple levels.
Finding a popular matching among the maximum weight matchings has also a similar algorithm as shown in Csáji et al. \cite{csaji2024popular}.

Sadly, if the preferences of the agents may contain ties, a popular matching may not exist and finding a popular matching or reporting that none exist is NP-hard \cite{cseh2017popularNPh}. 
In the case of arbitrary graphs, i.e. stable roommates instances, finding a popular matching was shown to be NP-hard even with strict preferences by Gupta et al. \cite{gupta2021popular} and Faenza at al. \cite{faenza2019popular} independently. 

To cope with this computational obstacle, Huang and Kavitha \cite{huang2017popularity} studied popular mixed-matchings in the marriage and roommmates settings. They showed that a popular mixed-matching always exists and gave an algorithm to find a maximum weight popular mixed-matching in polynomial-time by using an LP-formulation for the problem. This notion of popular mixed-matchings were also studied by Brandt and Bullinger \cite{brandt2022finding}. We introduce a new version of popularity for fractional-matchings, which we call popular fractional-matchings. We show that popular mixed-matchings may not be popular fractional-matchings, but our algorithm can find a popular fractional-matching that has maximum size among both the popular fractional- and popular mixed-matchings.

\subsection{Our Contributions}
In this paper we extend the $\frac{3}{2}$-approximation algorithm for \smti\ for roommates instances. A nice application for this result (somewhat surprisingly) arises from the problem of assigning residents to hospitals in the presence of couples. This is a notoriously difficult problem algorithmically, which is NP-hard even in extremely restricted cases \cite{biro2014hospitals}. However, it was shown very recently by Csáji et al. \cite{csaji2023couples}, that if we pose some natural restrictions on the couple's preferences, then we can have polynomial time algorithms that find stable matchings by modifying the capacities by at most 1 in the process (which leads to a so-called near-feasible stable matching). The novel idea of the algorithm is to reduce the problem to the stable roommates problem with appropriate gadgets and round off a weakly stable half-matching in the obtained instance. Therefore, given that the hospitals and the residents may have ties in their rankings, using algorithms that find larger weakly stable half-matchings in stable roommates instances can be utilized to find near-feasible stable solutions that match more residents, which is a very desirable feature. 

We also extend the algorithms for finding maximum size popular matchings and popular matchings among the maximum weight matchings (in the case of strict preferences). Since popular matchings may not exist in the roommates model, and even deciding their existence is NP-hard, we propose a natural definition of popularity among fractional-matchings. This is somewhat similar to a previous notion of popular mixed-matchings \cite{huang2017popularity}, however, we show that a mixed-matching may not be popular with respect to our notion. We show that there always exists a maximum size popular fractional-matching that is half-integral. We also give an algorithm that finds such a matching. This happens to be also a popular mixed-matching which has maximum size among the popular mixed-matchings. Finally, we give an algorithm that finds a half-integral matching that is popular among the maximum weight fractional-matchings. 

Our results are based on a simple, but crucial observation. If one translates the algorithms of Király \cite{kiraly2012linear} and Kavitha \cite{kavitha2014size} that have multiple possible levels for each agent such that they make parallel copies of different types of each edge instead that the two sides rank in an (almost) opposite order, then it seems essential that we can have a partition $U\cup W$ of the vertices such that each edge has one vertex from both. However, it turns out that we only need that for each edge, its parallel copies are ranked in an (almost) opposite order by its two vertices, we do not need to have copy types that are globally ranked in an (almost) opposite way by the two sides. This crucial observation helps us to overcome the difficulties of extending these algorithms for roommates instances.

\subsection{Paper Structure}
We start by defining the problems we study in more detail in Section \ref{sec:prelim}. Then, we describe the meta-algorithm in Section \ref{sec:alg}, which serves as a basis for all three of our algorithms. The approximation algorithm for \srti\ is described in Section \ref{sec:maxsrti}, the algorithm for finding a maximum size popular fractional-matching in Section \ref{sec:maxpop} and the algorithm for finding a popular half-integral matching among the maximum weight fractional-matchings in Section \ref{sec:popmax}. 
Finally, we conclude in Section \ref{sec:conc}.

\section{Preliminaries}\label{sec:prelim}

We investigate matching markets, where the set of agents with the possible set of contracts is given by an arbitrary graph $G=(V;E)$, where vertices correspond to agents and edges correspond to possible, mutually acceptable contracts. We allow the graph $G$ to have parallel edges, i.e. multiple contracts between the same two agents. For each agent $v\in V$, let $E(v)$ denote the edges that are incident to $v$.

We say that a function $M:E\to \mathbb{R}_{\ge 0}$ is a \textit{fractional-matching}, if $\sum_{e\in E(v)}M(e)\le 1$ for each $v\in V$. We say that $M$ is a \textit{half-matching}, if additionally $M(e)\in \{ 0,\frac{1}{2},1\}$ for each $e\in E$. We say that $v$ is \textit{saturated} by $M$, if it holds that $\sum_{e\in E(v)}M(e)= 1$, otherwise $v$ is \textit{unsaturated}. Similarly, an edge $e$ is \textit{saturated} by $M$, if $M(e)=1$, otherwise it is unsaturated. For a vertex $v$, let $E_M(v)$ denote the edges with positive $M$-value incident to $v$, that is $E_M(v) = \{ e\in E(v)\mid M(e) > 0\}$. 

For a number $x$, let $x^+=\max \{ 0,x\}$. 

\subsection{Weakly Stable and $\gamma$-stable Matchings}
We assume that for each agent $v$, there is a \textit{preference function} $p_v:E(v) \to \mathbb{R}_{\ge 0}$, which defines a weak ranking $\succeq_v$ over the incident edges of $v$. We emphasize that we assume that the agents (vertices) rank their incident edges instead of the other agents, because we allow parallel edges in our model, representing multiple types of contracts between two given agents. We also assume that $p_v(\emptyset )\le 0$, which denotes that an agent always weakly prefers to be matched to any acceptable partners rather than being unmatched. Such an instance is called a \textit{stable roommates instance with ties and incomplete preferences}. If the $\succeq_v$ induced preferences are strict, then we refer to it as a \textit{stable roommates instance.} %Furthermore, for each vertex $v$, there is also an integral capacity $q(v)\in \mathbb{N}$.

For a vertex $v$ and a fractional-matching $M$, we define $p_v(M):=\min_{e\in E_M(v)}p_v(e)$, if $v$ is saturated, otherwise $p_v(M):=p_v(\emptyset )$. 

%By abuse of notation, we sometimes refer to the partner of $v$ in $M$ as $M(v)$ too. 
We say that an edge $e=(u,v)$ \textit{blocks} a fractional-matching $M$, if $M(e)<1$ and (i) $p_u(e)>p_u(M)$ and (ii) $p_v(e)>p_v(M)$, that is, both are either unsaturated or have a worse partner with positive weight in $M$. 

A fractional-matching $M$ is called \textit{weakly stable} or just \textit{stable}, if there is no blocking edge to $M$. The problem of finding a maximum size weakly stable matching in a bipartite graph with weak preferences, called \smti\ is a well studied problem that is NP-hard even to approximate within some constant, but admits a linear time $\frac{3}{2}$-approximation \cite{Kiraly13}. Here, we consider the problem of finding a maximum-size weakly stable fractional-matching in arbitrary graphs, which we call {\sc maximum size stable roommates problem with ties and incomplete preferences}, or \srti\ for short.

Instead of weak stability, we define a more general notion, mirroring $\gamma$-stability for bipartite graphs \cite{csaji2023simple}. Suppose that for each edge $e=(u,v)$, there are two numbers $\gamma_e^u>0$ and $\gamma_e^v>0$ given. We say that a fractional-matching $M$ is \textit{$\gamma$-min stable}, if there is no blocking edge $e=(u,v)$ such that $p_u(e)-p_u(M)\ge \gamma_e^u$ and $p_v(e)-p_v(M)\ge \gamma_e^v$ holds. Similarly, a fractional-matching is \textit{$\gamma$-max stable}, if there is no blocking edge $e=(u,v)$ such that $p_u(e)-p_u(M)\ge \gamma_e^u$ or $p_v(e)-p_v(M)\ge \gamma_e^v$ holds. %The case of $\Delta$-stabilites with a set $F\subseteq E$ of free edges corresponds to the very special case, when $\gamma_e^v=\infty$, if $e\in F$ and $\gamma_e^v=\Delta$ otherwise. Hence, in our model, we allow very different types of conditions for each edge to block, independently from each other, which can incorporate many other special properties of certain applications. For example, if one thinks about job markets, and assumes that the underlying preferences are in correspondence with the salaries of the positions, there may be many other aspects of a workplace that make it desirable. Hence, for each agent and each different company and position, the increase in the salary that would make a job offer good enough for the applicant to switch, might be different.  Similary, depending on a company's existing employees and a new agent, it may depend on the specific skills of a new agent that are relevant for the company, how large of a salary within a contract would be worth it for the company to switch from an existing employee to the new one in a blocking pair. 

%Putting it all together, we define a general model, which incorporates all of the previously discussed ones. Suppose we are given a set $C\subseteq U\cup W$ of critical vertices, a set $E_c\subseteq E$ of critical edges and $\gamma_e^v$ values for each pair $(e,v)\in E\times (U\cup W)$ such that $v\in e$. We say that a matching $M$ is \cminstab, if $M$ is critical and there is no blocking edge $e=(u,w)$, such that $M\setminus \{ (u,M(u)),(M(w),w)\} \cup \{ e\}$ is still critical, and $\min \{ p_u(e)-p_u(M(u))-\gamma_e^u,p_w(e)-p_w(M(w))-\gamma_e^w\} \ge 0$ both hold. Otherwise, we call such a blocking edge a \textit{$c\gamma$-min blocking edge}. Clearly, this generalizes all three concepts: $\Delta$-min stability, free edges and critical vertices. Similarly we can define \cmaxstab\  matchings by replacing $\min \{ p_u(e)-p_u(M(u))-\gamma_e^u,p_w(e)-p_w(M(w))-\gamma_e^w\} \ge 0$ with $\max \{ p_u(e)-p_u(M(u))-\gamma_e^u,p_w(e)-p_w(M(w))-\gamma_e^w\} \ge 0$.

In order to be able to solve both problems with the same algorithm, we go one step further and consider a common generalization of the above two problems. Instead of one $\gamma_e^u$ value, we have values $0<\gamma_e^u < \delta_e^u$ for each vertex-edge pair. We say that an edge $e=(u,v)$ \textit{\cgblocks} a fractional-matching $M$, if either $\min \{ p_u(e)-p_u(M)- \gamma_e^u, p_v(e)-p_v(M)- \delta_e^v\} \ge 0$ or $\min \{ p_u(e)-p_u(M)- \delta_e^u,p_v(e)-p_v(M)-\gamma_e^v\} \ge 0$ holds. We say that a matching $M$ is \textit{\cgstab}, if no edge \cgblocks\ $M$.

If the $\gamma_e^v$ values are sufficiently small, then this corresponds to $\gamma$-max stability, and when the $\gamma_e^v$ and $\delta_e^v$ values are sufficiently close, then it corresponds to $\gamma$-min stability. Furthermore, if both the $\gamma_e^v$ and $\delta_e^v$ are sufficiently small, then we get back weak stability. 

Now we define the optimization problem we investigate, called \textsc{Maximum  $\gamma$-stable fractional-matching with ties and incomplete preferences} abbreviated as \usrti\ for short. 
\pbDef{\usrti}{
A graph $G=(V;E)$, $p_v:E(v)\to \mathbb{R}_{\ge 0}$ preference valuations for each $v\in V$, numbers $0<\gamma_e^v<\delta_e^v$ for each pair $(e,v)\in E\times V$ such that $v\in e$.
}{
A maximum size \cgstab\ fractional-matching $M$}

The special case of finding a maximum size weakly stable fractional-matching is called \srti.

%The main result of the paper is a simple $\frac{3}{2}$-approximation algorithm for \usrti.

\subsection{Popular Matchings}

For our popular matching related problems, we assume that the $p_v(e)$ values are different for each vertex $v\in V$, so they induce a strict preference order $\succ_v$ over its adjacent edges. Also, we assume $e\succ_v \emptyset$ for any adjacent edge $e$ to $v$.

Let $e,f$ be two elements of $E\cup \{ \emptyset \}$. An agent $v$ can compare these two elements and cast a vote, based on which one he prefers more according to $\succ_v$. Hence, we define $\vote_v(e,f)=+1$, if $e\succ_v f$, $\vote_v(e,f)=0$, if $e=f$ and $\vote_v(e,f)=-1$, if $f\succ_v e$. 

For integral matchings, popularity is defined as follows. $M$ is \textit{popular}, if $$\Delta (M,N)=\sum_{v\in V}\vote_v(M(v),N(v))\ge 0 $$ for any matching $N$, where $M(v)$ and $N(v)$ denote the edge $v$ receives in $M$ or $N$ respectively, if there is any, otherwise they are defined to be $\emptyset$. Hence, a matching $M$ is popular, if it does not lose in a head-to-head comparison to any other matching. 

We extend this definition in a natural way to fractional-matchings, which mirrors the definition of popularity for the case of many-to-many matching instances \cite{brandl2016popular}. 

Let $M,N $ be two fractional-matchings. We denote the fractional-matching defined by $\max \{M(e)-N(e),0\}$ for $e\in E$ as $(M-N)^+$. 
Agent $v\in V$ now compares the differences $\mplus$ and $\nplus$ by pairing the edges of $E_{\mplus}(v)$ with $E_{\nplus}(v)$ with some appropriate weights. This is done with the help of a function $\phi_v \colon \big( E(v) \cup  \{ \emptyset \} \big) \times \big( E (v) \cup  \{ \emptyset \} \big) \to \mathbb{R}_{\ge 0}$, such that 
\begin{enumerate}
    \item $\sum_{f\in E (v)\cup \{ \emptyset \} }\phi_v (e,f) =\mplus (e)$ for each $e\in E (v)$,
    \item $\sum_{f\in  E (v)\cup \{ \emptyset \}  }\phi_v (f,e) =\nplus (e)$ for each $e\in E(v)$,
    \item $\sum_{f \in  E (v)\cup \{ \emptyset \}}\phi_v (f,\emptyset) =(\sum_{e\in E(v)}M(e)-N(e))^+$,
    \item $\sum_{f \in  E (v)\cup \{ \emptyset \}}\phi_v (\emptyset,f)= (\sum_{e\in E(v)}N(e)-M(e))^+$.
\end{enumerate}

The function $\phi_v$ is called a \textit{feasible pairing for} $v$ if it satisfies these condition. 

It follows also that $\phi_v(e,e)=0$ for any $e\in E (v) \cup \{ \emptyset \}$, as either $\mplus(e)=0$ or $\nplus(e)=0$ and also either $(\sum_{e\in E(v)}M(e)-N(e))^+=0$ or $(\sum_{e\in E(v)}N(e)-M(e))^+=0$.

 We say that $\phi = \{ \phi_v \mid v\in V \}$ is a \textit{feasible pairing}, if $\phi_v$ is a feasible pairing for $v$, for each $v\in V$.

For a feasible pairing $\phi$, we define $$\vote_v(M,N,\phi_v) = \sum_{(e,f)\in (E(v)\cup \{ \emptyset \})^2}  \phi_v(e,f) \cdot \vote_v(e,f).$$

Finally, we let $$\Delta (M,N,\phi )=\sum_{v\in V} \vote_v(M,N,\phi_v)$$ and $$\Delta (M,N) = \min \big\{ \Delta (M,N,\phi) \bigm| \text{$\phi$ is a feasible pairing} \big\}$$.

We say that an fractional-matching $M$ is \textit{popular}, if $\Delta (M,N )\ge 0$ for any fractional-matching $N$, i.e. if no other fractional-matching can beat $M$ in a head-to-head election with respect to any feasible pairing.

We study the following optimization problem of finding a maximum size popular fractional-matching.

\pbDef{\maxpri}{
A graph $G=(V;E)$ with strict preference orders $\succ_v$ for each $v\in V$.
}{
A maximum size popular fractional-matching $M$}

\subsection{Popular Mixed-Matchings and a Common Generalization}

Let $M$ be a fractional-matching. Then, by a result of Balinski \cite{balinski1965integer} about the half-integrality of the matching polytope, we have that $M=\sum_{i=1}^k\lambda_iM_i$, where $0<\lambda_i\le 1$, $\sum_{i=1}^k\lambda_i =1$ and each $M_i$ is half-integral. 

We say that the fractional-matching $M$ is a \textit{popular mixed-matching}\cite{huang2017popularity}, if it holds that for any 
fractional-matching $N$, we have that 
{\scriptsize
$$\sum\limits_{v\in V}[\sum\limits_{e,f\in E(v)}M(e)N(f)\vote_v(e,f)+\sum\limits_{e\in E(v)}M(e)(1-\sum\limits_{f\in E(v)}N(f))\vote_v(e,\emptyset ) +\sum\limits_{f\in E(v)}N(f)(1-\sum\limits_{e\in E(v)}M(e))\vote_v(\emptyset ,f)]\ge 0.$$
}

It is shown by Huang and Kavitha \cite{huang2017popularity} using the half-integrality of the matching polytope that this is equivalent to only requiring it for half-integral matchings $N$, and if the instance is bipartite, only for integral matchings $N$.

\begin{example}
    We give an example to show that a popular mixed-matching may not be a popular fractional matching. We have a bipartite instance with agents $u_1,u_2,u_3$ on one side and $w_1,w_2$ on the other. The preferences are $w_1\succ_{u_i} w_2$ for $i=1,2$, only $w_2$ for $u_3$, and $u_1\succ_{w_1} u_2$, $u_1\succ_{w_2} u_2\succ_{w_2} u_3$. Take the half-matching $M$ with $M(u_1,w_1)=M(u_1,w_2)=M(u_2,w_1)=M(u_2,w_2)=\frac{1}{2}$ and $M(e)=0$ otherwise. It is straightforward to verify that $M$ is a popular mixed-matching, by comparing it to every integral matching $N$. 
    
    However, it is not a popular fractional-matching. Take $N$ with $N(u_1,w_1)=1, N(u_2,w_2)=N(u_3,w_2)=\frac{1}{2}$ and $N(e)=0$ otherwise. Then, the feasible pairings are unique for each agent and we have that $u_1,w_1$ both have votes $-\frac{1}{2}$, $u_2,w_2$ both have votes $+\frac{1}{2}$ and $u_3$ has vote $-\frac{1}{2}$ so $\Delta (M.N)<0$. 
\end{example}

We relax the axioms of feasible pairings to allow comparisons that arise from the definition of popular mixed-matchings $-$ note that here an edge $e$ with $M(e)=N(e)=\frac{1}{2}$ may be paired with different edges, which is not allowed in a feasible pairing. 

We say that $\phi\colon \big( E(v) \cup  \{ \emptyset \} \big) \times \big( E (v) \cup  \{ \emptyset \} \big) \to \mathbb{R}_{\ge 0}$ is a \textit{sensible-pairing} between fractional-matchings $M$ and $N$, if 
\begin{enumerate}
    \item\label{sens1} $\sum_{f\in E (v)\cup \{ \emptyset \} }\phi_v (e,f) =M (e)$ for each $e\in E (v)$,
    \item\label{sens2} $\sum_{f\in  E (v)\cup \{ \emptyset \}  }\phi_v (f,e) =N (e)$ for each $e\in E(v)$,
    \item \label{sens3}$\sum_{f \in  E (v)\cup \{ \emptyset \}}\phi_v (\emptyset,f)>0$ $\Longrightarrow$  $\sum_{e\in E(v)}M(e)<1$,
    \item \label{sens4}$\sum_{f \in  E (v)\cup \{ \emptyset \}}\phi_v (f,\emptyset)>0$ $\Longrightarrow$ $\sum_{e\in E(v)}N(e)<1$,
    \item \label{sens5}$\phi_u(e,e)=\phi_v(e,e)$ for each $e=(u,v)\in E$. 
\end{enumerate}

Note that by changing the last axiom to $\phi_u(e,e,)=\phi_v(e,e)=\min \{ M(e),N(e)\}$, changing 3. back to the stronger requirement
  $\sum_{f \in  E (v)\cup \{ \emptyset \}}\phi_v (\emptyset,f)= (\sum_{e\in E(v)}N(e)-M(e))^+$, and changing 4. to $\sum_{f \in  E (v)\cup \{ \emptyset \}}\phi_v (f,\emptyset) =(\sum_{e\in E(v)}M(e)-N(e))^+$, we get back the definition for a feasible pairing. Hence, any feasible pairing is sensible. For a pairing $\phi_v$ defined by $\phi_v (e,f)=M(e)N(f)$, $\phi_v (e,\emptyset )=M(e)(1-\sum_{f\in E(v)}N(f))$, $\phi_v (\emptyset ,f)=(1-\sum_{e\in E(v)}M(e))N(f)$ for $e=(u,v)$, which is the pairing in the definition of popular mixed-matchings, we have that all axioms are clearly satisfied. 

We say that a fractional-matching $M$ is \textit{extra-popular}, if $\Delta (M,N,\phi)\ge 0$ for any fractional matching $N$ and any sensible-pairing $\phi$ between $M$ and $N$.
\begin{remark}
It immediately follows that an extra-popular fractional-matching is both a popular fractional-matching and a popular mixed-matching. 
\end{remark}

We say that a fractional-matching $M$ is \textit{barely-defendable}, if for any other fractional-matching $N$, there exists some sensible pairing $\phi$ such that $\Delta (M, N,\phi )\ge 0$. 

\begin{remark}
    We also have that both popular mixed-matchings and popular fractional-matchings are barely-defendable. 
\end{remark}

\subsection{Popular Maximum Weight Matchings and Popular Critical Matchings}

Let us suppose that we have some weight function $\w:E\to \mathbb{R}$ on the edges.
By restricting the candidate fractional-matchings to only the maximum weight fractional-matchings, we can define a new version of popularity. We say that a maximum weight fractional-matching $M$ is a \textit{popular maximum weight fractional-matching}, if we have $\Delta (M,N,\phi)\ge 0$ for any other maximum weight fractional-matching $N$ and feasible pairing $\phi$. We call the problem of finding such a matching \popmax.

\pbDef{\popmax}{
A graph $G=(V;E)$ with strict preference orders $\succ_v$ for each $v\in V$ and a weight function $\w:E\to \mathbb{R}$.
}{
A popular maximum weight fractional-matching $M$.}

In our algorithm, we reduce this problem to a slightly easier problem, where we consider critical fractional-matchings. To define these, let us suppose we are given a set $\C\subseteq V$ of critical vertices. We say that a fractional-matching is \textit{critical}, if it saturates every critical vertex, i.e. if $\sum_{e\in E(v)}M(e)=1$ for all $v\in \C$. 

Finally, we say that a critical fractional-matching is a \textit{popular critical fractional-matching}, if $\Delta (M,N,\phi)\ge 0$ for any critical fractional-matching $N$ and feasible pairing $\phi$. 

\pbDef{\popcrit}{
A graph $G=(V;E)$ with strict preference orders $\succ_v$ for each $v\in V$ and a set of critical vertices $\C \subseteq V$.
}{
A popular critical fractional-matching $M$.}

\section{The Meta-Algorithm}\label{sec:alg}

As it turns out, similarly to the bipartite case, all these problems can be solved with very similar algorithms, that have the same outline. We call this the \textit{meta-algorithm}, which we describe now.

1. Create a stable roommates instance $I'=(G',\succ')$ with strict preferences (and parallel edges - note that this is also called Stable Multiple Activities in the literature \cite{cechlarova2005stable}) by making parallel copies $(e_i)_{i\in k}$ of each edge $e$ and creating strict preferences over the created edges.  

2. Run the extension of Tan's algorithm by Cechlárová et al.\cite{cechlarova2005stable} to obtain a stable half-matching $M'$ in the new instance $I'$. 

3. Take the projection $M$ of $M'$ to $I$, which is defined as $M(e):=\sum_{i}M'(e_i)$.

The second and third steps will be the same for all three algorithms. The only difference between them is in the first step, that is, how do we create parallel edges and how do we create strict preference orders over them.

\subsection{Maximum Size $\gamma$-Stable Fractional-Matching}\label{sec:maxsrti}

We start with our $\frac{3}{2}$-approximation algorithm for \usrti.

Let $V=\{ v_1,\dots, v_n\}$ be an indexing of the vertices.

\paragraph{The fine-tuned first step.}
For each edge $e=(v_i,v_j)\in E$, $i<j$ we create four parallel edges $e^1,e^2,e^3,e^4$. 
For $v_i$, $e^1$ is considered the best copy of $e$, while $e^2$ is second best, $e^3$ is third and $e^4$ is last. For $v_j$, $e^4$ is the best copy of $e$, $e^3$ is second, $e^2$ is third and $e^1$ is last.
 
We create the strict preferences of a vertex $v$ as follows. For any edges $e,f$, the worst copy of $f$ (which may be either $f^1$ or $f^4$) is always ranked worse than any best, second or third copy of $e$. Furthermore, the second best copy of $f$ is better than the best copy of $e$ if and only if $p_v(f)\ge p_v(e)+\gamma_f^v$ and the third copy of $f$ is better than the best copy of $e$ if and only if $p_v(f)\ge p_v(e)+\delta_f^v$. Within the best and last copies, the edges are ranked according to a strict extension of $v$'s original weak preferences. 
This can be obtained in the following way. For a vertex $v\in V$ and $e=(u,v)$ let $a(e)$ be the best, $b(e)$ be the second, $c(e)$ be the third and $d(e)$ be the last copy. Set 
$p_v(a(e)) = p_v(e)$, $p_v(b(e)) = p_v(e)-\gamma_e^v$ and $p_v(c(e)) = p_v(e)-\delta_e^v$. Then, for each vertex, these new (possibly negative) preference numbers define a weak order, in which we break ties in a way such that among edges with the same value the $c$ copies are best, the $b$ copies are second and the $a$ copies are last. Finally, the worst copies are ranked strictly last in the order they are in $p_v()$, breaking the ties arbitrarily. 

We start by showing that we can restrict ourselves to only consider half-matchings.

\begin{lemma}
    There always exists a half-integral optimal solution to \usrti.
\end{lemma}
\begin{proof}
Let $M$ be an arbitrary maximum size $\gamma$-stable fractional-matching. Using the fact that the convex hull of fractional-matchings form a half-integral polyhedron \cite{balinski1965integer}, we get that $M=\sum_i\lambda_iM_i$, where $\sum_i \lambda_i = 1$, $\lambda_i> 0$ and $M_i$ is a half-matching for each $i$. It follows that there exists at least one $M_i$ such that $\sum_{e\in E}M_i(e)\ge \sum_{e\in E}M(e)$. We show that $M_i$ is $\gamma$-stable. 

Suppose that $e=(u,v)$ $\gamma$-blocks $M_i$. If $u$ (or $v$) is unsaturated in $M_i$ then it is also unsaturated in $M$. 
Also, we have that $E_{M_i}(u)\subseteq E_M(u)$ and $E_{M_i}(v)\subseteq E_M(v)$. Therefore, we get that $p_v(M_i)\ge p_v(M)$ and $p_u(M_i)\ge p_u(M)$. 
This implies that $e$ $\gamma$-blocks $M$, a contradiction. 

It follows that there exists an optimal half-matching to \usrti.
\end{proof}

Finally, we show that our algorithm outputs a $\frac{3}{2}$-approximation.

\begin{theorem}\label{thm:maxsrti}
\usrti\ can be $\frac{3}{2}$-approximated in polynomial time.
\end{theorem}
\begin{proof}
Let $M$ be the output of the algorithm and $M'$ be the pre-image of $M$ in $I'$.

\begin{claim}
The output half-matching $M$ is \cgstab.

\end{claim}
\begin{claimproof}
Suppose for contradiction that there is a \cgblocking\ edge $e=(v_i,v_j)$ to $M$. Suppose by symmetry that $p_{v_i}(e)\ge p_{v_i}(M)+\gamma_e^{v_i}$ and $p_{v_j}(e)\ge p_{v_j}(M)+\delta_e^{v_j}$. Take the second best copy of $e$ for $v_i$, which is the third best copy of $e$ for $v_j$. This copy of $e$ blocks $M'$ by the definition of the rankings, contradiction. 
\end{claimproof}

\begin{claim}\label{claim:maxsrti}
    For any \cgstab\ half-matching $N$, it holds that $\sum_{e\in E}N(e)\le \frac{3}{2}\sum_{e\in E}M(e)$.
\end{claim}
\begin{claimproof}
Create a bipartite graph $G''=(V,V',E'')$, where $V'=\{ v_i'\mid v_i\in V\}$ is a copy of $V$ and for each $(v_i,v_j)\in E$, we have edges $(v_i,v_j')\in E''$ and $(v_i',v_j)\in E''$. 

Each half-matching $N$ defines a matching $N''$ in $G''$ as follows. %(By a result of Tan\cite{tan1991necessary}, we know that this consists of disjoint value $1$ edges and cycles of value $\frac{1}{2}$ edges).
If $N(e)=1$ for $e=(v_i,v_j)$, then we add $(v_i,v_j')$ and $(v_i',v_j)$ to $N''$. If we have a cycle of value $\frac{1}{2}$ edges on vertices $\{ v_{i_1},\dots , v_{i_k}\} $, then we add $(v_{i_1},v_{i_2}'), (v_{i_2},v_{i_3}'),\dots, (v_{i_k},v_{i_1}')$. Finally, if we have a path of value $\frac{1}{2}$ edges on vertices $\{ v_{i_1},\dots , v_{i_k}\} $, then we add $(v_{i_1},v_{i_2}'), (v_{i_2},v_{i_3}'),\dots, (v_{i_{k-1}},v_{i_k}')$.

Let $N$ be a maximum size $\gamma$-stable half-matching. 
Take the matchings $N'',M''$ corresponding to $N$ and $M$. 

Suppose for the contrary that $\sum_{e\in E}N(e)> \frac{3}{2}\sum_{e\in E}M(e)$. Then, we have that $|N''|>\frac{3}{2}|M''|$. This implies that either there is a component in $N''\triangle M''$ that has a single $N''$-edge or a component with two $N''$-edges and a single $M''$-edge. 

In the first case, we get that there is an edge $(v_i,v_j)$ such that $v_i$ and $v_j$ are both unsaturated in $M$, contradicting the $\gamma$-stability of $M$. In the second case, let the component be $\{ e''=(v_{i_1},v_{i_2}'),f''=(v_{i_2}',v_{i_3}),g''=(v_{i_3},v_{i_4}')\} $, where $f''=(v_{i_2}',v_{i_3})$ is the edge of $M''$. We have that $i_1,i_2,i_3$ are distinct and so are $i_2,i_3,i_4$, but it may happen that $i_1=i_4$. 

Since $v_{i_1}$ and $v_{i_4}'$ are unmatched by $M''$, $v_{i_1}$ and $v_{i_4}$ are unsaturated in $M'$. Since the best copy of $e$ for $v_{i_2}$ does not block $M'$, $v_{i_2}$ is saturated with first, second or third copies. Since the best copy of $g$ for $v_{i_3}$ does not block $M'$, $v_{i_3}$ is saturated with first, second or third copies. Hence, either the second or the third best copy of $f=(v_{i_2},v_{i_3})$ for $v_{i_2}$ has positive weight in $M'$. Suppose by symmetry it is the second best copy. Then, for the same reason, we must have $p_{v_{i_2}}(e)\ge p_{v_{i_2}}(f)+\gamma_e^{v_{i_2}}$ and $p_{v_{i_3}}(e)\ge p_{v_{i_3}}(g)+\delta_e^{v_{i_3}}$. This contradicts the $\gamma$-stability of $N$.
\end{claimproof}
\end{proof}

\subsubsection{Maximum Size Weakly Stable Fractional-Matching}

For completeness, we describe an even simpler algorithm for the special case \srti, i.e. the case, when there are no $\gamma$ numbers, each agent $v$ just has a preference valuation $p_v()$, inducing a weak preference order $\succeq_v$ ($e\succ_v f \Leftrightarrow p_v(e)>p_v(f)$ and $e\succeq_{v} f \Leftrightarrow p_v(e)\ge p_v(f)$) and our goal is to find a maximum size weakly stable half-matching.

Then, in the first step of the algorithm, we do the following. For each $e=(v_i,v_j), i<j$, we create three parallel edges $e^i,e^0,e^j$. For $v_i$, $e^i$ is best, $e^0$ is middle and $e^j$ is worst, and for $v_j$, $e^j$ is best, $e^0$ is middle and $e^i$ is worst. Then, the preference orders for $v_i\in V$ are created from $\succeq_{v_i}$ such that for each $e\in E(v_i)$, $e^i\succ e^0$, and $f^0\succ e^i$ if and only if $f$ is strictly preferred to $e$. Within the same copy types, the edges are ranked according to an arbitrary strict extension of $\succeq_{v_i}$. The copies $e^j$, $j\ne i$ are ranked last in the order $\succeq_v$, breaking the ties in an arbitrary way. An easy way to create these preferences is to go through $\succeq_{v_i}$, and in each tie $[e_1,..,e_k]$ just substitute it with $e^i_1\succ \cdots e_k^i\succ e_1^0\succ \cdots e_k^0$, then append to worst copies of the edges strictly after them by breaking the ties arbitrarily.

To give an example, let $\succeq_{v_i} = e\succ [f,g] \succ h$. Then, we get $e^i\succ e^0 \succ f^i\succ g^i \succ f^0 \succ g^0 \succ h^i \succ h^0 \succ e^{j_1} \succ f^{j_2} \succ g^{j_3} \succ h^{j_4}$, where $e=(v_i,v_{j_1}), f=(v_i,v_{j_2}), g=(v_i,v_{j_3}), h=(v_i,v_{j_4})$. 

\begin{theorem}\label{thm:maxsrtisimple}
\srti\ can be $\frac{3}{2}$-approximated in polynomial time even with the simpler algorithm.
\end{theorem}
\begin{proof}
Let $M$ be the output of the simpler algorithm and $M'$ be the pre-image of $M$ in $I'$.

\begin{claim}
The output half-matching $M$ is weakly stable.

\end{claim}
\begin{claimproof}
Suppose for contradiction that there is a blocking edge $e=(v_i,v_j)$ to $M$. Then, the copy $e^0$ blocks $M'$ by the definition of the rankings, contradiction. 
\end{claimproof}

\begin{claim}\label{claim:maxsrti}
    For any weakly stable half-matching $N$ it holds that $\sum_{e\in E}N(e)\le \frac{3}{2}\sum_{e\in E}M(e)$.
\end{claim}
\begin{claimproof}
Let $N$ be a maximum size weakly stable half-matching. 

Create a bipartite graph $G''=(V,V',E'')$, and matchings $N'',M''$ as in Theorem \ref{thm:maxsrti}.

Suppose for the contrary that $\sum_{e\in E}N(e)> \frac{3}{2}\sum_{e\in E}M(e)$. Then, we have that $|N''|>\frac{3}{2}|M''|$. This implies that either there is a component in $N''\triangle M''$ that has a single $N''$-edge or a component with two $N''$-edges and a single $M''$-edge. 

In the first case, we get that there is an edge $(v_i,v_j)$ such that $v_i$ and $v_j$ are both unsaturated in $M$, contradicting the weak stability of $M$. In the second case, let the component be $\{ e''=(v_{i_1},v_{i_2}'),f''=(v_{i_2}',v_{i_3}),g''=(v_{i_3},v_{i_4}')\} $, where $f''=(v_{i_2}',v_{i_3})$ is the edge of $M''$, i.e. $M(f)>0$. We have that $i_1,i_2,i_3$ are distinct and so are $i_2,i_3,i_4$, but it may happen that $i_1=i_4$. 

Since $v_{i_1}$ and $v_{i_4}'$ are unmatched by $M''$, $v_{i_1}$ and $v_{i_4}$ are unsaturated in $M'$. Since $e^{i_2}$ does not block $M'$, only $f^{i_2},f^0$ can have positive value in $M'$. Since $g^{i_3}$ does not block $M'$, only $f^{i_3},f^0$ can have positive value in $M'$. Hence, we have that $M'(f^0)>0$. Then, since $e^{i_2}$ and $g^{i_3}$ does not block $M'$, we must have $f\succ_{v_{i_2}} e$ and $f\succ_{v_{i_3}} g$, contradicting the stability of $N$. 
\end{claimproof}
\end{proof}

\subsection{Maximum Size Popular Fractional-Matching}\label{sec:maxpop}

In this section we solve the maximum size popular fractional-matching problem, \maxpri\ for short. Recall that for this problem, we assumed that we have strict preference orders $\succ_v$ for each agent $v\in V$.

\paragraph{The fine-tuned first step.} We describe the fine-tuned first step of the meta-algorithm.
For each edge $e=(v_i,v_j)$, $i<j$, we create two parallel edges $e_a,e_b$. Agent $v_i$ considers $e_a$ the good copy and $e_b$ the bad one, while $v_j$ considers $e_b$ the good and $e_a$ the bad one. 
The new strict preferences are created such that for any vertex $v_i$, it first ranks the good copies in their original order and strictly after then the bad copies in their original order. 

\begin{theorem}\label{thm:maxpri}
    \maxpri\ is solvable in polynomial time. Furthermore, the output of the algorithm is half-integral, extra-popular and has maximum size even among the barely-defendable fractional-matchings. 
\end{theorem}
\begin{proof}
Take the output half-matching $M$ of the algorithm. Take any fractional-matching $N$ and also an arbitrary sensible-pairing $\phi$ between $M$ and $N$.

%Consider a bipartite graph $G_0=(A,B,E_0)$, where each edge $e$ with $M(e)>N(e)$ has a corresponding vertex $v_e$ in $A$ and each edge $e$ with $N(e)>M(e)$ has a vertex $v_e$ in $B$. Also, we have a vertex corresponding to $\emptyset$ for both $A$ and $B$. Let $E_0$ consist of those $(v_e.v_f)$ pairs for which we have $\phi_0(e,f):=\phi (e,f)>0$ (note that $e$ or $f$ may be $\emptyset$, but not both). In case we have parallel edges in $G$, then we have as many parallel edges between $v_e$ and $v_f$ as the number of vertices $v$ with $\phi_v(e,f)>0$.

We decompose $\phi$ as follows.
First, we define $\phi^0$ by setting $\phi_v^0 (e,e)=0$ and $\phi_v^0(e,f)=\phi_v (e,f)$ for $e\ne f$. We also set $M^0(e)=M(e)-\phi_v(e,e)=M(e)-\phi_u(e,e)$ and $N^0(e)=N(e)-\phi_v(e,e)$ for all $e\in E$. We can do this, as $\vote_v(e,e)$ is always 0, so this does not change $\Delta (M,N, \phi)$, i.e. $\Delta (M,N,\phi )=\Delta (M^0,N^0,\phi^0)$. We also must still have that for each $e=(v_i,v_j)$, 
\begin{enumerate}
\item \label{item1}$\sum\limits_{f\in E(v_i)\cup \{ \emptyset \}}\phi^0_{v_i}(e,f) = \sum\limits_{f\in E(v_j)\cup \{ \emptyset \}}\phi^0_{v_j}(e,f)=M^0(e)$ and
\item \label{item2} $\sum\limits_{f\in E(v_i)\cup \{ \emptyset \}}\phi^0_{v_i}(f,e) = \sum\limits_{f\in E(v_j)\cup \{ \emptyset \}}\phi^0_{v_j}(f,e)=N^0(e)$, 
\end{enumerate}
because of the axiom $\phi_{v_i}(e,e)=\phi_{v_j}(e,e)$ on sensible-pairings.

Let $k=0$.
Create a graph $H^k=(W^k,E^k)$, where each $e=(v_i,v_j)\in E$ with $N^k(e)>0 $ receives two vertices $e_N^i,e_N^j$ and each with $M^k(e)>0$ receives vertices $e_M^i,e_M^j$. We also add a vertex $v^{\emptyset}$ corresponding to $\emptyset$. Then, we define the edges of $H^k$ such that $(e_M^i,f_N^j)$ is an edge if and only if $i=j$ and $\phi^k_{v_i}(e,f)>0$. We also add edges $(e_M^i,e_M^j)$ and $(e_N^i,e_N^j)$ for each $e=(v_i,v_j)$. Finally, we add an edge $(e_M^i,v^{\emptyset})$ if and only if $\phi^k_{v_i}(e,\emptyset)>0$ and an edge $(e_N^i,v^{\emptyset})$ if $\phi_{v_i}(\emptyset , e)>0$. Let $F^k=\{(e_M^i,e_M^j)\mid e\in E, \; M^k(e)>0\} \cup \{ (e_N^i,e_N^j)\mid e\in E, \; N^k(e)>0$\}. 

To proceed with the proof, we need the following claim.

\begin{claim}\label{claim:aux}
    Let $H=(W,E\cup F)$ ($E\cap F=\emptyset$) be a multigraph with $|E\cup F|\ge 1$, where $F$ is an almost perfect matching that covers all but one vertex $v_0\in W$ and every vertex $v\in W\setminus \{ v_0\}$ has degree at least one in $E$. Then, it holds that either there is an $E-F$-alternating closed walk in $H$, or an almost $E-F$-alternating closed walk, which alternates at every vertex except $v_0$. 
\end{claim}
\begin{claimproof}
    We prove this by induction on $|W|$. Clearly, as $F$ forms an almost perfect matching in $H$, where only $v_0$ is unmatched, we have that $|W| = 2m+1$ for some $m\in \mathbb{Z}$. If $|W|=1$, then $W=\{ v_0 \}$, so by $|E\cup F|\ge 1$, there is a loop edge $v_0v_0$ and the statement follows.  

    Suppose we know the statement for $|W|=2m+1$ and take the case, where $|W|=2m+3\ge 3$. Then, $F\ne \emptyset$. Let $(u,v)\in F$. If there is an edge $(u,v)\in E$ too, then we are done, since we found a closed $E-F$-alternating cycle. Otherwise, as $u$ and $v$ have degree at least one in $E$, we can create nonempty sets $U\subseteq W$ and $V\subseteq W$, such that each vertex in $U$ is connected to $u$ in $E$ and each vertex in $V$ is connected to $v$ in $E$. 

    Take the graph $H'$, which is the restriction of $H$ to $W\setminus \{ u,v\}$, but for each $u'\in U'$ and $v'\in V'$, we add an edge $(u',v')$. We call these edges virtual edges. Let $F'$ be the restriction of $F$ to $H'$ and $E'$ be the other edges. For this graph, we have that $|W'|=2m+1$, $F'$ forms an almost perfect matching, and every $v\in W'\setminus \{ v_0\}$ has degree at least one in $E'$ (if a vertex $w$ was only adjacent to $u$ or $v$, then we created at least one new edge incident to $w$). Since $E'\ne \emptyset$, as $|V'|\ge 1,|U'|\ge 1$, we have by induction that there is an $E'-F'$-alternating closed walk in $H'$ or an almost $E'-F'$-alternating closed walks, which alternates except from $v_0$. From this, we can create an $E-F$-alternating closed walk or an almost $E-F$-alternating closed walk in $H'$ by substituting each virtual edge $(u',v')$ with $(u',u),(u,v),(v,v'($.
\end{claimproof}

Observe that $F^k$ is an almost perfect matching in $H^k$ covering all vertices but $v^{\emptyset}$ and each vertex apart from $v^{\emptyset}$ is adjacent to at least one edge from $E^k\setminus F^k$. The latter holds as we have (and will maintain) properties \ref{item1} and \ref{item2}, meaning that for each $e=(v_i,v_j)$ that $\sum\limits_{f\in E(v_i)\cup \{ \emptyset \}}\phi^k_{v_i}(e,f) = \sum\limits_{f\in E(v_j)\cup \{ \emptyset \}}\phi^k_{v_j}(e,f)=M^k(e)$ and $\sum\limits_{f\in E(v_i)\cup \{ \emptyset \}}\phi^k_{v_i}(f,e) = \sum\limits_{f\in E(v_j)\cup \{ \emptyset \}}\phi^k_{v_j}(f,e)=N^k(e)$. 
Hence, using Claim~\ref{claim:aux} for $H^k=(W^k,E^k\setminus F^k \cup F^k)$ and $v_0=v^{\emptyset}$ we get that, if $\phi^k$ is not the all zero function yet (so $E^k\ne \emptyset$ by properties \ref{item1}, \ref{item2} and the construction of $H^k$), then either there is an $F^k-E^k\setminus F^k$ alternating closed walk in $H^k$ or an almost $F^k-E^k\setminus F^k$ alternating closed walk, which alternates at every vertex except $v^{\emptyset}$. 
%To show this, we create a graph $\widetilde{H^k}$, with the same set of vertices as $H^k$, but we create parallel edges from each edge as follows. For each vertex $w\in W^k$, let $d_{E^k\setminus F^k}(w)$ be the degree of $w$ in $(W^k,E^k\setminus F^k)$. Let $x$ be the smallest common multiple of $\{ d_{E^k\setminus F^k}(w)\mid w\in W^k\}$. Then, for an edge $(e_M^i,e_N^i)\in E^k\setminus F^k$ we create $x/d_{E^k\setminus F^k}($

Let this closed walk be $C^k$. Also observe that by the construction of $H^k$, the edges $C^k\cap F^k$ alternate between $(e_M^i,e_M^j)$ type and $(e_N^i,e_N^j)$ type edges. Hence, $C^k\cap F^k$ gives a walk $P^k$ in $G$, where $P^k=(e_0,f_1,e_1,\dots, f_{j_k},e_{j_k})$, with $M^k(e_i)>0$, $N^k(f_i)>0$ and $\phi_{v_{i-1}}(e_{i-1},f_i)>0, \phi_{u_i}(e_i,f_i)>0$ and it may happen that $e_0=\emptyset$, $e_{j_k}=\emptyset$ or $e_0=e_{j_k}$ 
 (note that an edge can have $N^k(e)>0$ and $M^k(e)>0$ simultaneously). 

 Intuitively, what we do in the following paragraphs is decreasing the $M^k,N^k$ and $\phi^k$ values according to the $e_i,f_i$ edges and $(e_{i-1},f_i),(e_i,f_i)$ comparisons that are given by the closed walk $C^k$, until one new $\phi^k(e,f)$ becomes 0 (while also being careful about the possibly different multiplicities of these in $C^k$).
 
 Let $e_i=(u_i,v_i)$ and $f_i=(v_{i-1},u_i)$. Let $\alpha_{i-1,i}$ be the multiplicity of the edge of $E^k$ corresponding to the comparison $(e_{i-1},f_i)$ in $C^k$ (the edge corresponding to a comparison $(e,f)$ is $(e_M^i,f_N^i)$), $\alpha_{i,i}$ be the multiplicity of the edge of $E^k$ corresponding to the comparison $(e_i,f_i)$ in $C^k$, $\alpha_i$ be the multiplicity of the edge of $F^k$ corresponding to the $M$-version of $e_i$ (the $M$-version of an edge $e$ is $(e_M^i,e_M^j)$) in $C^k$
 and finally $\beta_i$ be the multiplicity of the edge of $F^k$ corresponding to the $N$-version of $f_i$ in $C^k$. Similarly, let $\alpha_{0,0}$ be the multiplicity of the edge of $E^k$ corresponding to the comparison $(e_0,\emptyset)$ in $C^k$, if $v^{\emptyset}\in C^k$, but $e_0\ne \emptyset$ and $\alpha_{j_k,0}$ be the multiplicity of the edge of $E^k$ corresponding to the comparison $(e_{j_k},\emptyset)$ in $C^k$, if $v^{\emptyset}\in C^k$, but $e_{j_k}\ne \emptyset$.
 
 %Note that we can assume that the multiplicity of any edge of the form $(e_M^i,v^{\emptyset})$ or $(e_N^i,v^{\emptyset})$ is one, since if $C^k$ traverses $v^{\emptyset}$ more than one, then we can choose a segment between two incident appearance of $v^{\emptyset}$.
Set the weight of $P^k$ to be \begin{itemize}
    \item[--] $w_k=\min_{i\in [j_k]}\{\phi_{v_{i-1}}^k (e_{i-1},f_i)/\alpha_{i-1,i},\phi_{u_i}^k (e_i,f_i)/\alpha_{i,i}\}$, if $e_0=e_{j_k}$,
    \item[--]$w_k=\min_{i\in [j_k]}\{\phi_{v_{i-1}}^k (e_{i-1},f_i)/\alpha_{i-1,i},\phi_{u_i}^k (e_i,f_i)/\alpha_{i,i}, \phi_{u_0}(e_0,\emptyset)/\alpha_{0,0}\}$ if $e_0\ne \emptyset, e_{j_k}=\emptyset$, 
    \item[--] $w_k=\min_{i\in [j_k]}\{\phi_{v_{i-1}}^k (e_{i-1},f_i)/\alpha_{i-1,i},\phi_{u_i}^k (e_i,f_i)/\alpha_{i,i}, \phi_{v_{j_k}}(e_{j_k},\emptyset)/\alpha_{j_k,0}\}$, if $e_0=\emptyset, e_{j_k}\ne \emptyset$ and finally 
    \item[--] $w_k=\min_{i\in [j_k]}\{\phi_{v_{i-1}}^k (e_{i-1},f_i)/\alpha_{i-1,i},\phi_{u_i}^k (e_i,f_i)\alpha_{i,i}, \phi_{u_0}(e_0,\emptyset )/\alpha_{0,0}, \phi_{v_{j_k}}(e_{j_k},\emptyset)/\alpha_{j_k,0}\}$, if $\emptyset\ne e_0\ne e_{j_k}\ne \emptyset$. 
\end{itemize} 

Then, we let $k:=k+1$, $\phi_{v_{i-1}}^k(e_{i-1},f_i)=\phi_{v_{i-1}}^{k-1}(e_{i-1},f_i)-\alpha_{i-1,i}w_{k-1}$, $\phi^k_{u_i}(e_i,f_i)=\phi_{u_i}^{k-1}(e_i,f_i)-\alpha_{i,i}w_{k-1}$ ($i=1,\dots, j_{k-1}$) and $\phi_{v}^k(e,f)=\phi_{v}^{k-1}(e,f)$ otherwise. If $v^{\emptyset}\in C^{k-1}$ but $e_0\ne \emptyset$, then we let $\phi^k_{u_0}(e_0,\emptyset )=\phi^{k-1}_{u_0}(e_0,\emptyset) -\alpha_{0,0}w_{k-1}$, $\phi^k_{v_{j_{k-1}}}(e_{j_k},\emptyset) =\phi_{v_{j_{k-1}}}^{k-1}(e_{j_k},\emptyset) - \alpha_{j_{k-1},0}w_{k-1}$ and $\phi^k_v(e,\emptyset)=\phi^{k-1}_v(e,\emptyset)$ otherwise. %Also, $\phi_{v}^k(\emptyset , f) = \phi^{k-1}_{v}(\emptyset ,f)-w_{k-1}$, if $e_0=\emptyset, f=f_1, v=v_0$ or $e_{j_k}=\emptyset,f=f_{j_k},v=u_{j_k}$ and  $ \phi^k_{v}(\emptyset , e)=\phi^{k-1}_{v}(\emptyset, e)$ otherwise.

We let $M^{k}(e_i)=M^{k-1}(e_i)-\alpha_iw_{k-1}$ ($i=0,1,\dots, j_{k-1}$) and $M^k(e)=M^{k-1}(e)$ otherwise. We let $N^{k}(f_i)=N^{k-1}(f_i)-\beta_iw_{k-1}$ ($i=1,\dots, j_{k-1}$) and $N^k(f)=N^{k-1}(f)$ otherwise. By the choice of $w_{k-1}$, every value remains nonnegative. 

Finally, we define $H^k$ from $M^k,N^k$ and $\phi^k$ the same way and iterate this procedure until $\phi^k \equiv 0$. We can iterate it, since each vertex we keep in $H^k$ apart from $v^{\emptyset}$ remains incident to at least one edge from both $F^k$ and $E^k\setminus F^k$ as the value updates maintain $\sum_{f\in E(v)\cup \{ \emptyset \}}\phi_v^k(e,f) = \sum_{f\in E(u)\cup \{ \emptyset \}}\phi_u^k(e,f)=M^k(e)$ and $\sum_{f\in E(v)\cup \{ \emptyset \}}\phi_v^k(f,e) = \sum_{f\in E(u)\cup \{ \emptyset \}}\phi_u^k(f,e)=N^k(e)$. This procedure terminates in finitely many steps, as in each step at least one new $\phi^k_{v}(e,f)$ value becomes $0$. 

For each walk $P^l$, let $P^l=(e_0,f_1,e_1,f_2,e_2\dots, e_{j_l})$, where $M(e_i)>0$ and $N(f_i)>0$ for each $i$ and $e_0=e_{j_l}$, if $P^l$ is a closed walk. Note that $e_0,e_{j_l}$ can be $\emptyset$.

Recall that $e_i=(u_i,v_i)$ and $f_i=(v_{i-1},u_i)$. Of course, it may happen that $v_i=v_j$ or $u_i=v_j$ for some $i,j$, however, this does not matter for the proof. Note that depending on the above cases, we can have $u_0=\emptyset$, $v_{j_l}=\emptyset$ or $u_0=u_{j_l}, v_0=v_{j_l}$.

What is important however is that we have that no two adjacent edges are the same in any of the walks $P^l$ by the fact that $\phi_v^k(e,e)$ is always 0. This also means that $\vote_{v_{j-1}} (e_{j-1},f_j)+\vote_{u_j}(e_j,f_j))\in \{ -2,0,+2\}$, as the preferences are strict. 

Then, we have that $\Delta (M,N,\phi ) =\sum\limits_{l=1}^kw_l\cdot \sum\limits_{j=1}^{j_l}(\vote_{v_{j-1}} (e_{j-1},f_j)+\vote_{u_j}(e_j,f_j)+x(P^l))$, where $x(P^l)=\vote_{u_0}(e_0,\emptyset)+\vote_{v_{j_l}}(e_{j_l},\emptyset)$ if $e_0\ne e_{j_l}$ and 0 otherwise. Therefore, it is enough to show that for each walk $P^l$, $\sum\limits_{j=1}^{j_l}(\vote_{v_{j-1}} (e_{j-1},f_j)+\vote_{u_j}(e_j,f_j)+x(P^l))\ge 0$ and it is strictly larger than $0$, whenever we have $e_0=e_{j_l}=\emptyset$.

\begin{claim}
\begin{enumerate}
    \item\label{case1} If $e_i$ has a bad copy for $v_i$ with positive value in $M'$ and $e_{i+1}$ has a good copy for $v_{i+1}$ with positive value in $M'$, then $\vote_{v_i}(e_i,f_{i+1})+\vote_{u_{i+1}}(e_{i+1},f_{i+1})= +2$.
    \item \label{case2}If $e_i$ has a bad copy for $v_i$ with positive value in $M'$ or $e_{i+1}$ has a good copy for $v_{i+1}$ with positive value in $M'$, then $\vote_{v_i}(e_i,f_{i+1})+\vote_{u_{i+1}}(e_{i+1},f_{i+1})\ge 0$.
    \item \label{case3}If $e_i$ has a good copy for $v_i$ with positive value in $M'$ and $e_{i+1}$ has a bad copy for $v_{i+1}$ with positive value in $M'$, then $\vote_{v_i}(e_i,f_{i+1})+\vote_{u_{i+1}}(e_{i+1},f_{i+1})\ge -2$.
    \item\label{case4} If $e_0=\emptyset$, then $e_1$ has only a bad copy for $v_1$ with positive weight in $M'$ and $\vote_{u_1}(e_1,f_1)=+1$.
    \item \label{case5}If $e_{j_l}=\emptyset$, then $e_{j_l-1}$ has only a good copy for $v_{j_l-1}$ with positive weight in $M'$ and $\vote_{v_{j_l-1}}(e_{j_l-1},$ $f_{j_l})=+1$.
\end{enumerate}
\end{claim}
\begin{claimproof}
1. Suppose for the contrary that $v_i$ or $u_{i+1}$ prefers $f_{i+1}$ to $e_i$ or $e_{i+1}$ respectively. Note that both copies of $f_{i+1}$ are unsaturated in $M'$, because we have $M(e_i)>0$, $M(e_{i+1})>0$ and $f_{i+1}\notin \{ e_i,e_{i+1}\}$. 

By the assumption, both are matched with a bad copy with positive weight in $M'$. If $u_{i+1}$ prefers $f_{i+1}$, then the (unsaturated) bad copy of $f_{i+1}$ for $u_{i+1}$ (which is a good copy for $v_i$) blocks $M'$, otherwise the other copy blocks $M'$ by the definition of the extended rankings, contradiction.

2. In this case, at least one of $u_{i+1},v_i$ is matched with a bad copy with positive weight in $M'$, say it is $v_i$ (the other case is analogous). Hence, if both prefers $f_{i+1}$, then the bad copy of $f_{i+1}$ for $v_i$ (unsaturated in $M'$ for the same reasons) blocks $M'$, contradiction. 

3. Trivial.

4. If $e_0=\emptyset$, then we must have by axiom \ref{sens3} of sensible pairings that $v_0$ is unsaturated in $M'$, as $\phi_{v_0}(\emptyset , f_1)>0$. Hence, by the fact that the bad copy of $f_1$ for $v_0$ does not block $M'$, we get that $u_1$ is saturated with good copies in $M'$, so $e_1$ must have a bad copy with positive weight in $M'$ for $v_1$. Also, $u_1$ has to prefer each edge of $M'$ with positive weight even to the good copy of $f_1$, so $\vote_{u_1}(e_1,f_1)=+1$.

5. This case is very similar to \ref{case4}. We have that $u_{j_l}$ is unsaturated, so $v_{j_l-1}$ can only have good copies with positive weight in $M'$ and it must prefer $e_{j_l-1}$ to $f_{j_l}$, so $\vote_{v_{j_l-1}}(e_{j_l-1},f_{j_l})=+1$.
\end{claimproof}

We get by the above observations that for each $P^l$, case \ref{case1} happens at least as many times as case \ref{case3}, except when $e_0\ne \emptyset$, $e_{j_l}\ne \emptyset$ and $e_0\ne e_{j_l}$, in which case it can happen once more. However, then the two endpoints $u_0$ and $v_{j_l}$ vote with $+1$, i.e. $x(P^l)=+2$, so $\sum\limits_{j=1}^{j_l}(\vote_{v_{j-1}} (e_{j-1},f_j)+\vote_{u_j}(e_j,f_j)+x(P^l))\ge 0$ holds regardless.

Finally, if $e_0=e_{j_l}=\emptyset$, then case \ref{case1} happens one more time than case \ref{case3} by \ref{case4} and \ref{case5}, so we have that $\sum\limits_{j=1}^{j_l}(\vote_{v_{j-1}} (e_{j-1},f_j)+\vote_{u_j}(e_j,f_j)+x(P^l))> 0$. This proves that $M$ is extra-popular and that it has a maximum size even among the barely-defendable fractional-matchings. 
\end{proof}

\begin{remark}
    Observe that we did not use axiom~\ref{sens4} of sensible pairings anywhere in the proof of Theorem \ref{thm:maxpri}. Hence, the half-matching we get is extra-popular even in a slightly stronger sense we get by eliminating axiom~\ref{sens4} from the definition of sensible pairings. Also, it is largest among barely-defendable fractional-matchings even with the weaker definition which does not require axiom~\ref{sens4} from sensible pairings. 

    The reason for axiom~\ref{sens4} is that it is quite natural and symmetric to axiom~\ref{sens3}, and that we utilize it in the next section for the \popcrit\ algorithm. 
\end{remark}

\subsection{Popular Maximum Weight Fractional-Matching}\label{sec:popmax}

In this section, we consider the problem of finding a fractional-matching that is popular among the maximum weight fractional-matchings. 

%A fractional-matching $M$ is a \textit{maximum weight popular fractional-matching}, if $M$ is a maximum weight fractional-matching and $\Delta (M,N)\ge 0$ for any maximum weight fractional-matching $N$. 

Consider the LP for the maximum weight fractional-matching problem. 

\begin{equation*}
\begin{array}{rr@{}ll}
\text{maximize}  & \displaystyle\sum\limits_{e\in E} \w_{e}x_{e} &&\\
\text{subject to}& \displaystyle\sum\limits_{e\in E(v)}   x_{e} &\leq 1,  &v\in V\\
            & x_e& \ge 0 &  e\in E\\
\end{array}
\end{equation*}

The Dual LP is 

\begin{equation*}
\begin{array}{rr@{}ll}
\text{minimize}  & \displaystyle\sum\limits_{v\in V} y_{v}& &\\
\text{subject to}& y_u+y_v&\ge \w_e    & e\in E\\
 & y_v& \ge 0 &  v\in V\\
\end{array}
\end{equation*}

Let $y^*=(y^*_v)_{v\in V}$ be an optimal dual solution. 
By duality, we know that a fractional-matching $M$ given by $x=(x_e)_{e\in E}$ is optimal, if and only if we have (i) $y_u^*+y_v^*=\w(e)$ for each $e=uv$ with $x_e>0$ and (ii) for each $v$ with $y_v^*>0$ we have $\sum_{e\in E(v)}x(e)=1$. 

Let $E':=\{ e\in E\mid y_u^*+y_v^*=\w(e)\}$ and $\C := \{ v\in V\mid y_v^*>0\}$. Then, we have that the maximum weight fractional-matchings (so each possible candidate for a popular maximum weight fractional-matching) are exactly the ones that have positive value only on edges from $E'$ and saturate every $v\in \C$. Hence, we can restrict the set of edges of $G$ to $E'$, so from now on let $E:=E'$. 

This way, we reduced the problem to finding a \textit{popular critical fractional-matching 
}, which to recall is a fractional-matching $M$ that saturates every $v\in \C$ and $\Delta (M,N)\ge 0$ for any fractional-matching $N$ that saturates every vertex $v\in \C$. 

Therefore, now it is enough to solve \popcrit.

\paragraph{The fine-tuned first step.} We describe the fine-tuned first step of the meta-algorithm. Let $|\C| = s$. 
We keep each edge of $E$. These will be called the middle copies. If $e=(v_i,v_j)$ such that $v_i,v_j\notin \C$, we do not create additional parallel edges. If $v_i\in \C, v_j\notin \C$, we add copies $e^1,\dots, e^s$. These are considered worst copies for $v_i$ in the opposite order (i.e. $e^s$ is the absolute worst and $e^1$ is best among them) and best copies for $v_j$ in the opposite order (i.e. $e^s$ is the absolute best and $e^1$ is the worst among them). If $v_i\notin \C, v_j\in \C$, we add copies $e^{-1},\dots, e^{-s}$. These are considered worst copies for $v_j$ in the opposite order (i.e. $e^{-s}$ is the absolute worst and $e^{-1}$ is best among them) and best copies for $v_j$ in the opposite order (i.e. $e^{-s}$ is the absolute best and $e^{-1}$ is the worst among them). If $v_i,v_j\in \C$, then we do both of the above. 

Then, the preferences are created such that the best copies are best, then the middle copies and then the worst copies are worst. With respect to that, each vertex ranks the edges of the same type according to their original order. For a vertex $v$, we use the notation $B_j^v(e)$ to denote the $j$-th best copy of $e$ among the best copies, and $W_j^v(e)$ to denote the $j$-th worst copy of $e$ among the worst copies (both can be either $e^{j}$ or $e^{-j}$). We also use the notation $\lev_v(B_j^v(e))=j$, $\lev_v(e)=0$ and $\lev_v(W_j^v(e))=-j$, which is called the level of the given copy of $e$ for $v$.

\begin{theorem}\label{thm:popcrit}
    \popcrit\ is solvable in polynomial time. Furthermore, the algorithm outputs an extra-popular matching among the critical matchings.
\end{theorem}
\begin{proof}
    Take the output $M$ of the algorithm.

\begin{claim}
    $M$ saturates every vertex in $\C$. 
\end{claim}
\begin{claimproof}
    Suppose for the contrary that $M$ does not saturate some vertex $v\in \C$. Take a half-matching $N$ that does.  

    Construct a bipartite graph $G''=(V,V',E'')$ as in the proof of Theorem \ref{thm:maxsrti} and add the edges according to the half-matchings $N$ and $M$ the same way, which leads to matchings  $N''$ and $M''$. By the assumption on $N$, we know that there must be an $M''-N''$ alternating path that starts from some vertex $v_{i_1}$ or $v_{i_1}'$ with $v_{i_1}\in \C$ and contains more copies of vertices from $\C$. Suppose by symmetry it is $v_{i_1}'$. Take the edge $e_1''=(v_{i_1}',v_{i_2})\in N''$ (corresponding to the edge $e_1\in E$). Since not even the worst copy for $v_{i_1}$ of $e_1$ blocks $M'$, we get that $v_{i_2}$ is saturated in $M'$ with only $B_1^{v_{i_2}}$-type edges. Hence, we have some edge $f_1''=(v_{i_2},v_{i_3}')\in M''$ such that $v_{i_3}\in \C$, as $B_1^{v_{i_2}}(f_1)$ exist. By our assumption on the component, there is some edge $e_2''=(v_{i_3}',v_{i_4})\in N''$. For $v_{i_3}$, we have that $M'(W_1^{v_{i_3}}(f_1))>0$, so even $W_2^{v_{i_3}}(e_2)$ can only be dominated at $v_{i_4}$ in $M'$. Hence, $v_{i_4}$ is saturated in $M'$ with $B_j^{v_{i_4}}$ copies with $j\le 2$. Hence, we have an edge $f_2''=(v_{i_4},v_{i_5}')\in M''$ and that $v_{i_5}\in \C$. By the assumption on the component, there exists an edge $e_3''=(v_{i_5}',v_{i_6})\in N''$. By iterating this argument, we get edges $f_3'', e_4'', \dots, e_s'',f_s''$, such that each of them contains a vertex in $V'$ that corresponds to a critical vertex in $\C$. However, there are $s+1$ such vertices in $(e_1\cup f_1\cup \cdots \cup e_s\cup f_s)\cap V'$, a contradiction to $|\C |=s$. 
\end{claimproof}   

\begin{claim}
    $M$ is an extra-popular critical fractional-matching. 
\end{claim}
\begin{claimproof}
Let $N$ be an arbitrary critical fractional-matching and let $\phi$ be a sensible pairing between $M$ and $N$. 

We decompose $\phi$ to get walks $P^1,\dots, P^k$ the same way as in Theorem \ref{thm:maxpri}. Let $P^l=(e_0,f_1,e_1,$ $f_2,e_2\dots, e_{j_l})$, where $M(e_i)>0$, $N(f_i)>0$ and $\phi_{v_{i-1}}(e_{i-1},f_i)>0,\phi_{u_i}(e_i,f_i)>0$ for each $i$ and $e_0=e_{j_l}$, if $P^l$ is a closed alternating walk (i.e. if it is a union of alternating cycles), and otherwise $e_0=\emptyset$ or $e_{j_l}=\emptyset$ is possible. Here, $e_i=(u_i,v_i)$ and $f_i=(v_{i-1},u_i)$.

Then, we have that $\Delta (M,N,\phi ) =\sum\limits_{l=1}^kw_l\cdot \sum\limits_{j=1}^{j_l}(\vote_{v_{j-1}} (e_{j-1},f_j)+\vote_{u_j}(e_j,f_j)+x(P^l))$, where $x(P^l)=\vote_{u_0}(e_0,\emptyset)+\vote_{v_{j_l}}(e_{j_l},\emptyset)$ if $e_0\ne e_{j_l}$ and 0 otherwise.. Therefore, it is enough to show that for each walk $P^l$, $\sum\limits_{j=1}^{j_l}(\vote_{v_{j-1}} (e_{j-1},f_j)+\vote_{u_j}(e_j,f_j)+x(P^l))\ge 0$.

It may happen that $v_i=v_j$ or $u_i=v_j$ for some $i,j$, however, this does not matter for the proof. Similarly as in Theorem \ref{thm:maxpri}, we have that for any $i$, $f_{i}\notin \{ e_{i-1},e_i\}$.% Note that depending on the above cases, we can have $u_0=\emptyset$, $v_{j_l}=\emptyset$ or $u_0=u_{j_l}, v_0=v_{j_l}$.

\begin{claim}\label{claim:popmax-cases}
\begin{enumerate}
    \item\label{c1} If $e_i$ has a level $p$ copy for $v_i$ in $M'$ with positive value and $e_{i+1}$ has a level $q$ copy for $v_{i+1}$ with positive value in $M'$, then $q\le p+1$. 
    \item\label{c2} If $e_i$ has a level $p$ copy for $v_i$ in $M'$ with positive value and $e_{i+1}$ has a level $q$ copy for $v_{i+1}$ with positive value in $M'$ such that $q= p+1$, then $\vote_{v_i}(e_i,f_{i+1})+\vote_{u_{i+1}}(e_{i+1},f_{i+1})= +2$.
    \item\label{c3} If $e_i$ has a level $p$ copy for $v_i$ in $M'$ with positive value and $e_{i+1}$ has a level $q$ copy for $v_{i+1}$ with positive value in $M'$ such that $p=q$, then $\vote_{v_i}(e_i,f_{i+1})+\vote_{u_{i+1}}(e_{i+1},f_{i+1})\ge 0$.
    \item\label{c4} If $e_0=\emptyset$, then each copy of $e_1$ with positive value in $M'$ has level at most $0$ for $v_1$ and if there is one with level exactly $0$, then
    $\vote_{u_1}(e_1,f_1)=+1$.
    \item\label{c5} If $e_0\ne \emptyset$ and $e_0\ne e_{j_l}$, then each copy of $e_0$ with positive value in $M'$ has level at least $0$ for $v_0$.
    \item \label{c6}If $e_{j_l}=\emptyset$, then each copy of $e_{j_l-1}$ with positive value in $M'$ has level at least $0$ for $v_{j_l-1}$ and if there is one with level exactly $0$, then $\vote_{v_{j_l-1}}(e_{j_l-1},f_{j_l})=+1$.
    \item \label{c7}If $e_{j_l}\ne \emptyset$ and $e_{j_l}\ne e_0$, then each copy of $e_{j_l}$ with positive value in $M'$ has level at least $0$ for $v_{j_l}$.
\end{enumerate}
\end{claim}
\begin{claimproof}
First of all, observe that for any edge $f_i$, it is not saturated in $M$, because at least one of its endpoints have a different incident edge with positive value (and it cannot happen that $P_l=(e_0,f_1,e_1)$ with $e_0=e_1=\emptyset$ by the stability of $M'$). Hence, none of $f_i$'s copies are saturated in $M'$.

    1. Suppose that $q>p+1$. Then, the copy of $f_{i+1}$ that has level $p+1$ for $v_i$ blocks $M'$, because this copy is level $-(p+1)>-q$ for $u_{i+1}$, so both of them strictly prefer it to some copies of $e_i$ and $e_{i+1}$ with positive value in $M'$, contradiction.

    2. In this case, $u_{i+1}$ has a level $-(p+1)$ copy of $f_{i+1}$ with positive value in $M'$. Hence, if either $v_i$ or $u_{i+1}$ prefer $f_{i+1}$ to $e_i$ or $e_{i+1}$ respectively, then either the level $p$ or the level $p+1$ copy of $f_{i+1}$ for $v_i$ blocks $M'$, contradiction. 

    3. Suppose both $v_i$ and $u_{i+1}$ prefer $f_{i+1}$ to $e_i$ and $e_{i+1}$ respectively. Then, by the assumption, the level $p$ copy of $f_{i+1}$ for $v_i$ blocks $M'$, contradiction. 

     4. Suppose $e_0=\emptyset$. Then, $\phi (f_1,\emptyset )>0$, so $v_0$ is unsaturated in $M'$ by axiom~\ref{sens3} on sensible pairings. By the fact that the level $0$ copy (which must exist) of $f_1$ for $v_0$ does not block $M'$, we get that $u_1$ is saturated in $M'$ with copies that have level at least $0$. Hence, $e_1$ has level at most $0$ for $v_1$ in any copy that has positive value in $M'$. Also, if the level $0$ copy of $e_1$ has positive value in $M'$, then we must have that $e_1\succ_{u_1}f_1$, so $\vote_{u_1}(e_1,f_1)=+1$. 
    
    5. Suppose $e_0\ne \emptyset$ and $e_0\ne e_{j_l}$. This implies that $\phi_{u_0}(e_0 , \emptyset)>0$. By axiom~\ref{sens4} on sensible pairings, this implies that $\sum_{e\in E(u_0)}N(e)<1$, so $u_0\notin \C$, as both $M$ and $N$ saturate critical vertices. Hence, $u_0$ has only copies with level at least $0$ in $G'$, so $e_0$ has only copies with level at most $0$ for $v_0$ in $G'$ and thus also in $M'$.

    6. Suppose that $e_{j_l}=\emptyset$. Then, $u_{j_l}$ is unsaturated by $M'$ by axiom~\ref{sens3} of feasible pairings. By the fact that the level $0$ copy (which must exist) of $f_{j_l}$ for $u_{j_l}$ does not block $M'$, we get that $v_{j_l-1}$ is saturated in $M'$ with copies that have level at least $0$. Hence, $e_{j_l-1}$ has level at least $0$ for $v_{j_l-1}$ in any copy that has positive value in $M'$. Also, if the level $0$ copy of $e_{j_l-1}$ has positive value in $M'$, then we must have that $e_{j_l-1}\succ_{v_{j_l-1}}f_{j_l}$, so $\vote_{v_{j_l-1}}(e_{j_l-1},f_{j_l})=+1$. 

    7. Similarly as in case \ref{c5}, we get that $v_{j_l}\notin \C$. Therefore, any copy of $e_{j_l}$ in $G'$ have level at least $0$ for $v_{j_l}$, thus this is also true for the ones with positive value in $M'$. 
\end{claimproof}

From Claim \ref{claim:popmax-cases} cases \ref{c4} -\ref{c7} it follows that the occurrences of case \ref{c2} with $q=p+1$ are at least as many as the occurrences with $q<p$. Furthermore, we also have that $\vote_{v_0}(e_0,f_1)+\vote_{u_1}(e_1,f_1)\ge 0$ and $\vote_{u_{j_l}}(e_{j_l},f_{j_l})+\vote_{v_{j_l-1}}(e_{j_l-1},f_{j_l})\ge 0$, if $e_0=\emptyset$ or $e_{j_l}=\emptyset$ by cases \ref{c4} and \ref{c6}, so we can conclude that
$\sum\limits_{j=1}^{j_l}(\vote_{v_{j-1}} (e_{j-1},f_j)+\vote_{u_j}(e_j,f_j)+x(P^l))\ge 0$ (by case~\ref{c2} and $\vote_{v_{j-1}} (e_{j-1},f_j)+\vote_{u_j}(e_j,f_j)\ge -2$ even if $q<p$).
\end{claimproof}

Therefore, we have that $M$ is a popular maximum weight fractional-matching.
\end{proof}

As we observed before, we get the following corollary from Theorem \ref{thm:popcrit}
\begin{corollary}
    \popmax\ is solvable in polynomial time. 
\end{corollary}

\section{Conclusions}\label{sec:conc}
In this paper we have shown that many algorithms that rely on leveling the agents, or duplicating the edges can be extended even to roommates instances. 
We extended the $3/2$-approximation algorithm for maximum-size weakly stable matchings from bipartite graphs to arbitrary instances. This resolved a long-standing open question and has important applications, particularly in allocating junior doctors to hospitals in the presence of couples among the doctors.

Our central observation may open even more possibilities to utilize this technique even further, so it would be interesting to see other frameworks, where it helps to provide desirable solutions.

%\section{Acknowledgements}
%The author thanks Tamás Király, Kavitha Telikepalli and Yu Yokoi for fruitful discussions about popular fractional-matchings.
%The project was supported by the Hungarian Scientific Research Fund, OTKA, Grant No. K143858, by the Momentum Grant of the Hungarian Academy of Sciences, grant number 2021-2/2021 and by the Ministry of Culture and Innovation of Hungary from the National Research, Development and Innovation fund, financed under the KDP-2023 funding scheme (grant number C2258525).

\bibliographystyle{abbrv}
\bibliography{cit}

\begin{thebibliography}{10}

\bibitem{egres}
\url{http://lemon.cs.elte.hu/egres/open/Maximum_weakly_stable_matchings_in_graphs_without_odd_preference_cycles}.

\bibitem{nrmp}
\url{https://www.nrmp.org/}.

\bibitem{balinski1965integer}
M.~L. Balinski.
\newblock Integer programming: methods, uses, computations.
\newblock {\em Management science}, 12(3):253--313, 1965.

\bibitem{biro2014hospitals}
P.~Bir{\'o}, D.~F. Manlove, and I.~McBride.
\newblock The hospitals/residents problem with couples: Complexity and integer programming models.
\newblock In {\em Experimental Algorithms: 13th International Symposium, SEA 2014, Copenhagen, Denmark, June 29--July 1, 2014. Proceedings 13}, pages 10--21. Springer, 2014.

\bibitem{brandl2016popular}
F.~Brandl and T.~Kavitha.
\newblock Popular matchings with multiple partners.
\newblock {\em arXiv preprint arXiv:1609.07531}, 2016.

\bibitem{brandt2022finding}
F.~Brandt and M.~Bullinger.
\newblock Finding and recognizing popular coalition structures.
\newblock {\em Journal of Artificial Intelligence Research}, 74:569--626, 2022.

\bibitem{cechlarova2005stable}
K.~Cechl{\'a}rov{\'a} and V.~Val’ov{\'a}.
\newblock The stable multiple activities problem.
\newblock {\em IM Preprint, series A}, (1):2005, 2005.

\bibitem{csaji2023simple}
G.~Cs{\'a}ji.
\newblock A simple 1.5-approximation algorithm for a wide range of max-smti problems.
\newblock {\em arXiv e-prints}, pages arXiv--2304, 2023.

\bibitem{csaji2024popular}
G.~Cs{\'a}ji, T.~Kir{\'a}ly, K.~Takazawa, and Y.~Yokoi.
\newblock Popular maximum-utility matchings with matroid constraints.
\newblock {\em arXiv preprint arXiv:2407.09798}, 2024.

\bibitem{csaji2023approximation}
G.~Cs{\'a}ji, T.~Kir{\'a}ly, and Y.~Yokoi.
\newblock Approximation algorithms for matroidal and cardinal generalizations of stable matching.
\newblock In {\em Symposium on Simplicity in Algorithms (SOSA)}, pages 103--113. SIAM, 2023.

\bibitem{csaji2023couples}
G.~Cs{\'a}ji, D.~Manlove, I.~McBride, and J.~Trimble.
\newblock Couples can be tractable: New algorithms and hardness results for the hospitals/residents problem with couples.
\newblock {\em arXiv preprint arXiv:2311.00405}, 2023.

\bibitem{cseh2017popularNPh}
{\'A}.~Cseh, C.-C. Huang, and T.~Kavitha.
\newblock Popular matchings with two-sided preferences and one-sided ties.
\newblock {\em SIAM Journal on Discrete Mathematics}, 31(4):2348--2377, 2017.

\bibitem{smti1.5inapprox}
S.~Dudycz, P.~Manurangsi, and J.~Marcinkowski.
\newblock Tight inapproximability of minimum maximal matching on bipartite graphs and related problems.
\newblock In {\em Approximation and Online Algorithms: 19th International Workshop, WAOA 2021, Lisbon, Portugal, September 6--10, 2021, Revised Selected Papers}, pages 48--64. Springer, 2022.

\bibitem{faenza2019popular}
Y.~Faenza, T.~Kavitha, V.~Powers, and X.~Zhang.
\newblock Popular matchings and limits to tractability.
\newblock In {\em Proceedings of the Thirtieth Annual ACM-SIAM Symposium on Discrete Algorithms}, pages 2790--2809. SIAM, 2019.

\bibitem{gale1962college}
D.~Gale and L.~S. Shapley.
\newblock College admissions and the stability of marriage.
\newblock {\em The American Mathematical Monthly}, 69(1):9--15, 1962.

\bibitem{gardenfors1975match}
P.~G{\"a}rdenfors.
\newblock Match making: assignments based on bilateral preferences.
\newblock {\em Behavioral Science}, 20(3):166--173, 1975.

\bibitem{gupta2021popular}
S.~Gupta, P.~Misra, S.~Saurabh, and M.~Zehavi.
\newblock Popular matching in roommates setting is np-hard.
\newblock {\em ACM Transactions on Computation Theory (TOCT)}, 13(2):1--20, 2021.

\bibitem{halldorsson2002inapproximability}
M.~Halld{\'o}rsson, K.~Iwama, S.~Miyazaki, and Y.~Morita.
\newblock Inapproximability results on stable marriage problems.
\newblock In {\em LATIN 2002: Theoretical Informatics: 5th Latin American Symposium Cancun, Mexico, April 3--6, 2002 Proceedings 5}, pages 554--568. Springer, 2002.

\bibitem{yanigasawa2003improved}
M.~M. Halld{\'o}rsson, K.~Iwama, S.~Miyazaki, and H.~Yanagisawa.
\newblock Improved approximation of the stable marriage problem.
\newblock In {\em Algorithms-ESA 2003: 11th Annual European Symposium, Budapest, Hungary, September 16-19, 2003. Proceedings 11}, pages 266--277. Springer, 2003.

\bibitem{huang2017popularity}
C.-C. Huang and T.~Kavitha.
\newblock Popularity, mixed matchings, and self-duality.
\newblock In {\em Proceedings of the Twenty-Eighth Annual ACM-SIAM Symposium on Discrete Algorithms}, pages 2294--2310. SIAM, 2017.

\bibitem{irving2002stable}
R.~W. Irving and D.~F. Manlove.
\newblock The stable roommates problem with ties.
\newblock {\em Journal of Algorithms}, 43(1):85--105, 2002.

\bibitem{IMMM99}
K.~Iwama, D.~Manlove, S.~Miyazaki, and Y.~Morita.
\newblock Stable marriage with incomplete lists and ties.
\newblock In {\em Proc. 26th International Colloquium on Automata, Languages, and Programming (ICALP 1999)}, pages 443--452. Springer, 1999.

\bibitem{IMY07}
K.~Iwama, S.~Miyazaki, and N.~Yamauchi.
\newblock A 1.875-approximation algorithm for the stable marriage problem.
\newblock In {\em Proc. Eighteenth annual ACM-SIAM symposium on Discrete algorithms (SODA 2007)}, pages 288--297. SIAM, Philadelphia, 2007.

\bibitem{IMY08}
K.~Iwama, S.~Miyazaki, and N.~Yamauchi.
\newblock A ($2-c\frac{1}{\sqrt{n}}$)-approximation algorithm for the stable marriage problem.
\newblock {\em Algorithmica}, 51(3):342--356, 2008.

\bibitem{kavitha2014size}
T.~Kavitha.
\newblock A size-popularity tradeoff in the stable marriage problem.
\newblock {\em SIAM Journal on Computing}, 43(1):52--71, 2014.

\bibitem{Kiraly11}
Z.~Kir{\'a}ly.
\newblock Better and simpler approximation algorithms for the stable marriage problem.
\newblock {\em Algorithmica}, 60(1):3--20, 2011.

\bibitem{kiraly2012linear}
Z.~Kir{\'a}ly.
\newblock Linear time local approximation algorithm for maximum stable marriage.
\newblock In {\em Proc. Second International Workshop on Matching Under Preferences (MATCH-UP 2012)}, page~99, 2012.

\bibitem{Kiraly13}
Z.~Kir{\'a}ly.
\newblock Linear time local approximation algorithm for maximum stable marriage.
\newblock {\em Algorithms}, 6(3):471--484, 2013.

\bibitem{manlove2013algorithmics}
D.~Manlove.
\newblock {\em Algorithmics of matching under preferences}, volume~2.
\newblock World Scientific, 2013.

\bibitem{manlove2002hard}
D.~F. Manlove, R.~W. Irving, K.~Iwama, S.~Miyazaki, and Y.~Morita.
\newblock Hard variants of stable marriage.
\newblock {\em Theoretical Computer Science}, 276(1-2):261--279, 2002.

\bibitem{Mcdermid09}
E.~McDermid.
\newblock A 3/2-approximation algorithm for general stable marriage.
\newblock In {\em Proc. 36th International Colloquium on Automata, Languages, and Programming (ICALP 2009)}, pages 689--700. Springer, 2009.

\bibitem{critical-ties-approx}
M.~Nasre, P.~Nimbhorkar, and K.~Ranjan.
\newblock Critical relaxed stable matchings with two-sided ties.
\newblock {\em arXiv preprint arXiv:2303.12325}, 2023.

\bibitem{paluch2011faster}
K.~Paluch.
\newblock Faster and simpler approximation of stable matchings.
\newblock In {\em Proc. 9th International Workshop on Approximation and Online Algorithms (WAOA 2011)}, pages 176--187, 2011.

\bibitem{Paluch14}
K.~Paluch.
\newblock Faster and simpler approximation of stable matchings.
\newblock {\em Algorithms}, 7(2):189--202, 2014.

\bibitem{tan1991necessary}
J.~J. Tan.
\newblock A necessary and sufficient condition for the existence of a complete stable matching.
\newblock {\em Journal of Algorithms}, 12(1):154--178, 1991.

\bibitem{yanagisawa2007approximation}
H.~Yanagisawa.
\newblock Approximation algorithms for stable marriage problems.
\newblock {\em PhD thesis, Kyoto University, Graduate School of Informatics}, 2007.

\end{thebibliography}

\end{document}